\documentclass[12pt]{article}

\usepackage{graphics,graphicx,fullpage,natbib,multirow}
\usepackage{amsmath,amssymb,verbatim,epsfig}
\usepackage[dvipsnames,usenames]{color}
\usepackage[normalem]{ulem}
\usepackage{soul}
\usepackage{xcolor}
\usepackage[colorlinks,citecolor=blue,urlcolor=blue,linkcolor=blue]{hyperref}

\newtheorem{proposition}{Proposition}

\newtheorem{defin}{\bf Definition}

\newenvironment{proof}{\noindent{\bf Proof}}{$\diamond$}

\def\ga{\mbox{Ga}}

\def\be{\mbox{Be}}

\def\no{\mbox{N}}

\def\un{\mbox{Un}}

\def\P{\mbox{P}}

\def\pd{\mbox{PD}}

\def\d{\mbox{d}}

\def\bu{{\bf u}}

\def\bX{{\bf X}}

\def\bxi{{\boldsymbol{\xi}}}
\def\bzero{{\bf 0}}
\def\bone{{\bf 1}}
\def\simind{\stackrel{\mbox{\scriptsize{ind}}}{\sim}}
\def\simiid{\stackrel{\mbox{\scriptsize{iid}}}{\sim}}
\newcommand{\balpha}{\boldsymbol{\alpha}}

\begin{document}

\title{\bf Bayesian nonparametric mixtures of asymmetric Archimedean copulas}
\author{
{\sc Radu V. Craiu and Luis E. Nieto-Barajas} \\[2mm]
{\sl University of Toronto \& ITAM} \\[2mm]
{\small {\tt radu.craiu@utoronto.ca \&  luis.nieto@itam.mx }}
}
\date{}
\maketitle

\begin{abstract}
Archimedean copulas are exchangeable: a single generator parameter controls every pairwise association, and the copula is invariant under permutations of its arguments. Even when mixtures of Archimedean copulas are formulated, the exchangeability feature is inherited. Non exchangeable copulas can be defined through the Khoudraji transformation of an Archimedean copula. This relies in replacing the arguments by a power transformation and adding a factor based on the independence copula. The construction produces pairwise Kendall tau coefficients that depend on the power parameters and recovers the original Archimedean copula as a particular case. We propose a Bayesian nonparametric mixture model by mixing over the copula generator parameter $\theta$ and the parameters $\balpha$ that characterize the asymmetry. We use the Poisson--Dirichlet process as a mixing distribution, which groups observations by both strength and direction of dependence. We derive the copula density, the pairwise association measures, the tail coefficients, and the full posterior conditional distributions to perform a Bayesian analysis. We carry out simulation studies and real data analyses. The results show that the Khoudraji mixture improves greatly on, or agrees  with, the exchangeable fit. 
\end{abstract}

\vspace{0.2in} \noindent {\sl Keywords}: Archimedean copula, asymmetric copula, Bayesian nonparametrics, Khoudraji transform, mixture model, multivariate dependence.

\section{Introduction}
\label{sec:intro}

Let $\bX=(X_1,\ldots,X_d)$ be a continuous $d$-dimensional random vector with joint cumulative distribution function (henceforth, CDF) $F$ and marginal CDFs $F_j$.  By \cite{Sklar:1959} there is a unique copula $C:[0,1]^d\to[0,1]$ such that $F(x_1,\ldots,x_d)=C(F_1(x_1),\ldots,F_d(x_d))$.  Dependence measures such as Kendall's tau are functionals of $C$ alone \citep{nelsen:06}.

A leading parametric class is the Archimedean family
\begin{equation}
\label{eq:ac}
C(u_1,\ldots,u_d)=\psi^{-1}\bigl\{\psi(u_1)+\cdots+\psi(u_d)\bigr\},
\end{equation}
generated by a continuous, decreasing, and convex function $\psi:[0,1]\to[0,\infty]$ with $\psi(1)=0$ and $\psi(0)=\infty$ \citep{nelsen:06,mcneil&neslehova:09}. The copula $C_\psi$ is exchangeable in the sense that $C_\psi(u_{\sigma(1)},\ldots,u_{\sigma(d)})=C_\psi(u_1,\ldots,u_d)$ for every permutation $\sigma$, and every pairwise margin is the same bivariate Archimedean copula. Consequently, every pair $(X_j,X_k)$ shares the same Kendall tau $\tau_\theta$, a single number determined by the generator parameter $\theta$.

\cite{pan&al:25} mixed Archimedean copulas with a Poisson--Dirichlet process \citep{pitman&yor:97} on $\theta$.  The Bayesian nonparametric (BNP) mixture is more flexible than a single Archimedean copula because it can be multimodal and it clusters the sample in subgroups with different association parameters. However, every component remains exchangeable, so the model cannot distinguish $\tau_{12}$ from $\tau_{13}$. %
This work is motivated by the need to extend \cite{pan&al:25} to include non-exchangeable copula structures within a BNP mixture approach. Importantly, the new construction generalizes and contains the exchangeable case.

A standard device for breaking exchangeability, due to \cite{khoudraji:95} and studied by \cite{genest&al:98}, \cite{liebscher:08} and \cite{durante&salvadori:10}, is to combine an Archimedean copula with the independence copula after coordinate-wise power transforms.  The resulting family is no longer symmetric in its arguments, pairwise Kendall $\tau$'s may differ, and one coordinate can be closer to independence than another, while the original Archimedean copula and the independence copula remain nested.

This paper extends the Bayesian nonparametric mixture of symmetric Archimedean copulas \citep{pan&al:25} to the assymmetric Khoudraji case. This involves placing a Poisson--Dirichlet process prior on the pair $(\theta,\balpha)$, where $\theta$ is the generator parameter and $\balpha=(\alpha_1,\ldots,\alpha_d)\in[0,1]^d$ is the Khoudraji asymmetry vector. 

The rest of the paper is organised as follows. Section~\ref{sec:khoudraji} reviews the transform, derives the copula density, and records association and tail formulae. Section~\ref{sec:mixmodel} defines the Bayesian nonparametric mixture and the posterior sampler. Section~\ref{sec:numerical} describes numerical experiments that parallel those of \cite{pan&al:25}, including refits of their occupancy and red-wine analyses. Section~\ref{sec:conclusion} concludes.

\section{Khoudraji-transformed Archimedean copulas}
\label{sec:khoudraji}

\subsection{The construction}

Let $C_\theta$ be a $d$-dimensional Archimedean copula with generator $\psi_\theta$ and association parameter $\theta\in\Theta$, %
and let $\Pi(\bu)=\prod_{j=1}^d u_j$ denote the independence copula. For $\balpha=(\alpha_1,\ldots,\alpha_d)\in[0,1]^d$, the Khoudraji transform of $C_\theta$ is
\begin{align}
\nonumber
C_{\theta,\balpha}(\bu)
&=C_\theta\bigl(u_1^{\alpha_1},\ldots,u_d^{\alpha_d}\bigr)\,
\Pi\bigl(u_1^{1-\alpha_1},\ldots,u_d^{1-\alpha_d}\bigr) \\
\label{eq:khoudraji}
&=\psi_{\theta}^{-1}\bigl\{\psi_{\theta}(u_1^{\alpha_1})+\cdots+\psi_{\theta}(u_d^{\alpha_d})\bigr\}
\prod_{j=1}^d u_j^{1-\alpha_j}.
\end{align}
The product of two copulas evaluated in coordinate powers is itself a copula for every $\balpha\in[0,1]^d$ \citep{khoudraji:95,liebscher:08}. We denote by $c_{\theta,\balpha}$ the density of $C_{\theta,\balpha}$  when it exists.

\begin{proposition}[Special exchangeable cases]
\label{prop:special}
Let $C_{\theta,\balpha}$ be given by \eqref{eq:khoudraji}.
\begin{enumerate}
\item[(a)] If $\balpha=\bone$, then $C_{\theta,\balpha}=C_\theta$.
\item[(b)] If $\balpha=\bzero$, then $C_{\theta,\balpha}=\Pi$.
\item[(c)] If $\alpha_j=0$ for a single coordinate $j$, then $U_j$ is independent of $(U_k)_{k\neq j}$ under $C_{\theta,\balpha}$.
\item[(d)] The bivariate margin of $(U_j,U_k)$ is the bivariate Khoudraji transform of the bivariate Archimedean copula generated by $\psi_{\theta}$, with asymmetry parameters $(\alpha_j,\alpha_k)$. In particular $C_{\theta,\balpha}(u,v)\neq C_{\theta,\balpha}(v,u)$ whenever $\alpha_j\neq\alpha_k$, so the copula is not exchangeable.
\end{enumerate}
\end{proposition}

\begin{proof}
(a) and (b) are immediate from $C_\theta(\bu)\Pi(\bone)=C_\theta(\bu)$ and $C_\theta(\bone)\Pi(\bu)=\Pi(\bu)$.
(c) Substituting $u_j$ and $u_k=1$ for $k\neq j$ gives $C_{\theta,\balpha}(1,\ldots,u_j,\ldots,1)=u_j$, and the joint CDF factors are given by $u_j$ times a copula on the remaining coordinates because $C_\theta(\ldots,u_j^{0},\ldots)=C_\theta(\ldots,1,\ldots)$ does not depend on $u_j$ except through explicit power $u_j^{1-\alpha_j}=u_j$.
(d) The $(j,k)$ margin is
\[
C_{\theta,\balpha}(u,v,1,\ldots,1)
=C_\theta(u^{\alpha_j},v^{\alpha_k},1,\ldots,1)\,u^{1-\alpha_j}v^{1-\alpha_k}.
\]
Archimedean copulas satisfy $C_\theta(u,v,1,\ldots,1)=\psi_{\theta}^{-1}\{\psi_{\theta}(u)+\psi_{\theta}(v)\}$, the corresponding bivariate Archimedean copula. The resulting function of $(u,v)$ is therefore \eqref{eq:khoudraji} in dimension two with parameters $(\alpha_j,\alpha_k)$.  It is symmetric in $(u,v)$ if and only if $\alpha_j=\alpha_k$.
\end{proof}

Thus, the vector $\balpha$ interpolates between a fully exchangeable Archimedean copula and independence, coordinate by coordinate. The copula of \cite{pan&al:25} is the corner $\balpha=\bone$ of the present family.

\subsection{Bivariate density}

For posterior inference, we need the density of the copula. In dimension two write $\alpha,\beta\in[0,1]$ for the two asymmetry parameters and $C_\theta(u,v)$ for the bivariate Archimedean copula, with first partials 
\begin{equation}
\label{eq:hfun1}
C_\theta^{(u)}(u,v)=\frac{\partial}{\partial u}C_\theta(u,v)=\frac{\psi_{\theta}^{(1)}(u)}{\psi_{\theta}^{(1)}\{C_\theta(u,v)\}},\quad C_\theta^{(v)}(u,v)=\frac{\partial}{\partial v}C_\theta(u,v)=\frac{\psi_{\theta}^{(1)}(v)}{\psi_{\theta}^{(1)}\{C_\theta(u,v)\}}, 
\end{equation}
and cross second derivative
\begin{equation}
\label{eq:hfun2}
c_\theta(u,v)=\frac{\partial^2}{\partial u\partial v}C_\theta(u,v)=-\frac{\psi_{\theta}^{(1)}(u)\psi_{\theta}^{(1)}(v)\psi_{\theta}^{(2)}\{C_\theta(u,v)\}}{\bigl[\psi_{\theta}^{(1)}\{C_\theta(u,v)\}\bigr]^3},
\end{equation}
where $\psi_\theta^{(1)}(t)$ and $\psi_\theta^{(2)}(t)$ are the first and second derivatives of the generator $\psi_\theta(t)$ with respect to $t$, respectively. 
Table~\ref{tab:afam} shows five generators $\psi_\theta$ for the most common Archimedean families, together with their first two derivatives.

\begin{proposition}[Bivariate Khoudraji density]
\label{prop:dens2}
The copula \eqref{eq:khoudraji} in dimension two admits the density
\begin{align}
\label{eq:dens2}
c_{\theta,\alpha,\beta}(u,v)
&=
\alpha\beta\,c_\theta(u^\alpha,v^\beta)
+\alpha(1-\beta)\,v^{-\beta}\,C_{\theta}^{(u)}(u^\alpha,v^\beta)
\nonumber\\
&\qquad
+\beta(1-\alpha)\,u^{-\alpha}\,C_{\theta}^{(v)}(u^\alpha,v^\beta)
+(1-\alpha)(1-\beta)\,u^{-\alpha}v^{-\beta}\,C_\theta(u^\alpha,v^\beta),
\end{align}
for $(u,v)\in[0,1]^2$, provided $c_\theta$ exists.  If $\alpha=\beta=1$ then $c_{\theta,\alpha,\beta}=c_\theta$; if $\alpha=\beta=0$ then $c_{\theta,\alpha,\beta}=1$.
\end{proposition}

\begin{proof}
Write $C(u,v)=A(u,v)B(u)D(v)$ with $A=C_\theta(u^\alpha,v^\beta)$, $B=u^{1-\alpha}$ and $D=v^{1-\beta}$.  Differentiating the product twice and using $A_u=C_{\theta}^{(u)}(u^\alpha,v^\beta)\,\alpha u^{\alpha-1}$, $B_u=(1-\alpha)u^{-\alpha}$ together with the analogous formulae in $v$ produces four terms, which rearrange into \eqref{eq:dens2}.  The two exchangeable cases are Proposition~\ref{prop:special}(a)--(b).
\end{proof}

The four terms in \eqref{eq:dens2} correspond to the four ways of assigning each coordinate to the Archimedean factor or to the independence factor.  Substituting \eqref{eq:hfun1} and \eqref{eq:hfun2} expresses the density entirely in terms of $\psi_\theta$ and its first two derivatives, so the implementation of \cite{pan&al:25} extends by evaluating $C_\theta$, $C_{\theta}^{(u)}$, $C_{\theta}^{(v)}$ and $c_\theta$ at $(u^\alpha,v^\beta)$ and forming the linear combination \eqref{eq:dens2}.

In dimension $d>2$ the density exists whenever $C_\theta$ has a density \citep{mcneil&neslehova:09}.  It is the mixed partial of \eqref{eq:khoudraji} and can be evaluated by automatic differentiation or by the recursive expansion of \cite{liebscher:08}; the bivariate formula is the one required for the pairwise association calculations below and for the numerical experiments of Section~\ref{sec:numerical}.

\subsection{Pairwise Kendall tau}

Kendall's tau of a pair $(U,V)$ is
\begin{equation}
\label{eq:tau}
\tau_{\theta,\alpha,\beta}=4\int_0^1\int_0^1 C_{\theta,\alpha,\beta}(u,v)\,c_{\theta,\alpha,\beta}(u,v)\,\d u\,\d v-1,
\end{equation}
where $C_{\theta,\alpha,\beta}$ is any bivariate margin of Proposition~\ref{prop:special}(d). Unlike the Archimedean case, $\tau_{\theta,\alpha,\beta}$ depends on the pair through $(\alpha,\beta)$ as well as through $\theta$.

\begin{proposition}[Pairwise Kendall tau]
\label{prop:tau}
Write $\tau_\theta$ for Kendall's tau of $C_\theta$ as in Table~\ref{tab:afam}.  For the Khoudraji copula \eqref{eq:khoudraji},
\begin{enumerate}
\item[(a)] $\tau_{\theta,\alpha,\beta}=\tau_\theta$ if $\alpha=\beta=1$;
\item[(b)] $\tau_{\theta,\alpha,\beta}=0$ if $\alpha=0$ or $\beta=0$;
\item[(c)] if $\alpha=\beta\in[0,1]$, then $\tau_{\theta,\alpha,\alpha}$ interpolates continuously between $0$ (at $\alpha=0$) and $\tau_\theta$ (at $\alpha=1$);
\item[(d)] in general $\tau_{\theta,\alpha,\beta}$ is obtained from \eqref{eq:tau} with $C_{\theta,\alpha,\beta}$ and $c_{\theta,\alpha,\beta}$ given by \eqref{eq:khoudraji} and \eqref{eq:dens2}, and the map $(\alpha,\beta)\mapsto\tau_{\theta,\alpha,\beta}$ is continuous on $[0,1]^2$ for any valid $\theta\in\Theta$.
\end{enumerate}
\end{proposition}

\begin{proof}
(a) and (b) follow from Proposition~\ref{prop:special}.  For (c) and (d), $C_{\theta,\alpha,\beta}(u,v)$ is jointly continuous in $(\alpha,\beta,u,v)$ on $[0,1]^2\times(0,1)^2$ and bounded, so \eqref{eq:tau} is continuous in $(\alpha,\beta)$ for any valid $\theta\in\Theta$. The endpoint values are (a) and (b).
\end{proof}

\begin{table}[h]
\centering
\begin{tabular}{cccc} \hline \hline \\[-3mm]
Copula & $\psi_\theta(t)$ & $\psi_\theta^{(1)}(t)$ & $\psi_\theta^{(2)}(t)$ \\[1mm]  \hline \\[-3mm]
AMH & $\log\left\{\frac{1-\theta(1-t)}{t}\right\}$ & $\frac{\theta-1}{t\{1-\theta(1-t)\}}$ & $\frac{(1-\theta)(1-\theta+2\theta t)}{\left\{t(1-\theta(1-t))\right\}^2}$ \\[2mm]
CLA & $t^{-\theta}-1$ & $-\theta t^{-(\theta+1)}$ & $\theta(\theta+1) t^{-(\theta+2)}$ \\[2mm]
FRA & $-\log\left\{\frac{\exp(-\theta t)-1}{\exp(-\theta)-1}\right\}$ & $\frac{\theta\exp(-\theta t)}{\exp(-\theta t) - 1}$ & $\frac{\theta^2\exp(-\theta t)}{\{\exp(-\theta t) -1\}^2}$ \\[2mm]
GUM & $(-\log t)^{\theta}$ & $-\frac{\theta}{t}(-\log t)^{\theta-1}$ & $\frac{\theta}{t^2} (-\log t)^{\theta -1}+\frac{\theta(\theta-1)}{t^2}(-\log t)^{\theta -2}$ \\[2mm]
JOE & $-\log\left\{1-(1-t)^\theta\right\}$ & $\frac{-\theta(1-t)^{\theta-1}}{1-(1-t)^\theta}$ & $\frac{\theta(\theta-1)(1-t)^{\theta-2}+\theta(1-t)^{2\theta-2}}{\{1-(1-t)^\theta\}^2}$ \\[2mm]
\hline \hline
\end{tabular}
\caption{Common Archimedean families used as the inner copula $C_\theta$ in \eqref{eq:khoudraji}. Generator $\psi_{\theta}(t)$, together with its first $\psi_\theta^{(1)}(t)$ and second $\psi_\theta^{(2)}(t)$ derivatives. }
\label{tab:afam}
\end{table}

There is, in general, no closed form for $\tau_{\theta,\alpha,\beta}$ except at the corners of the unit square in $(\alpha,\beta)$. Numerical integration of \eqref{eq:tau}, or a Monte Carlo estimate from a sample of the copula, is therefore the operational formula.%

The essential modelling gain is that the Khoudraji copula for $d>2$ can now carry a vector of pairwise tau's. For $d=3$ and $\balpha=(\alpha_1,\alpha_2,\alpha_3)$, the three pairs $(\tau_{\theta,\alpha_1,\alpha_2},\tau_{\theta,\alpha_1,\alpha_3},\tau_{\theta,\alpha_2,\alpha_3})$ are functions of the same $\theta$ but of three different pairs of asymmetry parameters, and they coincide if and only if $\alpha_1=\alpha_2=\alpha_3=1$ (the Archimedean corner) or the corresponding $\alpha$'s vanish.

\subsection{Tail coefficients}

Write $\lambda_L=\lim_{u\downarrow 0}C(u,u)/u$ and $\lambda_U=\lim_{u\uparrow 1}\{1-2u+C(u,u)\}/(1-u)$ for the lower and upper tail dependence coefficients of a bivariate copula $C(u,v)$. For the five Archimedean families of Table~\ref{tab:afam} one has $\lambda_L=2^{-1/\theta}$ and $\lambda_U=0$ for Clayton with $\theta>0$; $\lambda_L=0$ and $\lambda_U=2-2^{1/\theta}$ for Gumbel and for Joe; and $\lambda_L=\lambda_U=0$ for Frank and AMH.

\begin{proposition}[Tail coefficients of the Khoudraji transform]
\label{prop:tail}
Let $C_{\theta,\alpha,\beta}$ denote \eqref{eq:khoudraji} in dimension two with $\lambda_L(\theta,\alpha,\beta)$ and $\lambda_U(\theta,\alpha,\beta)$ its corresponding lower and upper tail coefficients, respectively. 
\begin{enumerate}
\item[(a)] $\lambda_L(\theta,\alpha,\beta)=0$ unless $\alpha=\beta=1$, in which case the Archimedean coefficient is recovered.
\item[(b)] If the inner copula $C_\theta$ is Gumbel, the Khoudraji transform is an extreme-value copula with
\[
\lambda_U(\theta,\alpha,\beta)=\alpha+\beta-(\alpha^\theta+\beta^\theta)^{1/\theta},
\]
which is strictly positive whenever $\alpha,\beta>0$ and $\theta>1$, and equals $2-2^{1/\theta}$ at $\alpha=\beta=1$.
\item[(c)] If the inner copula $C_\theta$ is Clayton, Frank or AMH, then $\lambda_U(\theta,\alpha,\beta)=0$.  If it is Joe, upper-tail dependence is inherited in diluted form for $\alpha,\beta>0$.
\end{enumerate}
Thus lower-tail dependence (Clayton) requires both coordinates to be fully assigned to the Archimedean factor, whereas upper-tail dependence (Gumbel, Joe) survives as soon as both asymmetry parameters are positive.
\end{proposition}

\begin{proof}
Any copula satisfies $C(s,t)\le\min(s,t)$, so
\[
\frac{C_{\theta,\alpha,\beta}(u,u)}{u}
\le u^{\max(\alpha,\beta)+1-\alpha-\beta}
=u^{1-\min(\alpha,\beta)}.
\]
The right-hand side tends to $0$ unless $\min(\alpha,\beta)=1$, i.e. unless $\alpha=\beta=1$. This is (a).
\noindent
For (b), Gumbel is extreme-value with stable tail dependence function $\ell(x,y)=(x^\theta+y^\theta)^{1/\theta}$.  The Khoudraji transform of an extreme-value copula is extreme-value \citep{genest&al:98,durante&salvadori:10}, with Pickands function
\[
A_{\alpha,\beta}(w)=\bigl[(\alpha w)^\theta+\{\beta(1-w)\}^\theta\bigr]^{1/\theta}+(1-\alpha)w+(1-\beta)(1-w).
\]
The identity $\lambda_U=2\bigl\{1-A(1/2)\bigr\}$ yields the displayed formula.  For $\theta>1$ the map $x\mapsto x^\theta$ is strictly convex, so $(\alpha^\theta+\beta^\theta)^{1/\theta}<\alpha+\beta$ whenever $\alpha,\beta>0$, and therefore $\lambda_U>0$.  (c) follows from the corresponding Archimedean coefficients and Proposition~\ref{prop:special}(a).
\end{proof}

Asymmetry of the density survives for $\alpha\neq\beta$ even when both tail coefficients vanish: $c_{\theta,\alpha,\beta}(u,v)\neq c_{\theta,\alpha,\beta}(v,u)$, so probability mass concentrates towards one axis. In the following  Section~\ref{sec:mixmodel} we propose a mixture that can describe populations in which some clusters are exchangeable and heavy-tailed (Archimedean, $\balpha=\bone$) while others are asymmetric, with tail behaviour governed by Proposition~\ref{prop:tail}.

\section{BNP mixtures}
\label{sec:mixmodel}

\subsection{Model}

Let $\bxi=(\theta,\balpha)\in\Xi:=\Theta\times[0,1]^d$ collect the Archimedean parameter and the Khoudraji vector.  We mix the family $\{C_{\theta,\balpha}(\,\cdot\,):\bxi\in\Xi\}$ with a Poisson--Dirichlet process on $\Xi$, exactly as \cite{pan&al:25} mixed $\{C_\theta(\,\cdot\,):\theta\in\Theta\}$ on $\Theta$.  A probability measure $G$ has a Poisson--Dirichlet prior $G\sim\pd(a,b,G_0)$ with $a\in[0,1)$, $b>-a$ and centering measure $G_0$ on $\Xi$ when
\begin{equation}
\label{eq:pdp}
G(\cdot)=\sum_{k=1}^\infty \omega_k\delta_{\bxi_k}(\cdot),
\end{equation}
with stick-breaking weights $\omega_k$ as in \cite{pitman&yor:97} and atoms $\bxi_k\simiid G_0$.  The mixture copula is
\begin{equation}
\label{eq:mixC}
C(\bu)=\int C_{\bxi}(\bu)\,G(\d\bxi)=\sum_{k=1}^\infty\omega_k C_{\theta_k,\balpha_k}(\bu).
\end{equation}
Hierarchically, for observations $i=1,\ldots,n$,
\begin{align}
\nonumber
(U_{1i},\ldots,U_{di})\mid\theta_i,\balpha_i
&\simind c_{\theta_i,\balpha_i}(\bu_i),\\
\label{eq:bnpmix}
(\theta_i,\balpha_i)\mid G
&\simiid G,\\
\nonumber
G
&\sim\pd(a,b,G_0),
\end{align}
where $c_{\theta,\balpha}$ is the copula density associated to \eqref{eq:khoudraji}. 
Because $G$ is almost surely discrete, $\P(\bxi_i=\bxi_{i'})>0$, and the distinct values of $\bxi_i$ define clusters that share both a dependence strength $\theta$ and an asymmetry pattern $\balpha$.  Smaller $a$ and $b$ produce more ties.

The centering measure $G_0$ is defined through its density $g_0$ given by,
\begin{equation}
\label{eq:G0}
g_0(\theta,\balpha)=h(\theta)\prod_{j=1}^d \be(\alpha_j\mid c_\alpha,d_\alpha),
\end{equation}
where $h(\theta)$ is the centering density of the Archimedean parameter that varies according to its parameter space $\Theta$. Specifically, we take $\un(\theta\mid-1,1)$ for AMH, $\ga(\theta\mid c_\theta,d_\theta)$ for Clayton, $\ga(\theta-1\mid c_\theta,d_\theta)$ for Gumbel and Joe, and $\no(\theta\mid\mu_\theta,\lambda_\theta)$ for Frank.  We restrict Clayton to $\theta>0$: negative bivariate Clayton copulas have a singular boundary component and therefore are not represented by the absolutely continuous likelihood used here. Independent beta margins on the $\alpha_j$ cover the whole unit interval and include the Archimedean corner in the support (the prior mass near $\balpha=\bone$ is controlled by $c_\alpha,d_\alpha$). A more concentrated choice, such as $\be(\alpha_j\mid 8,1)$, encodes the belief that variables in the data are close to exchangeable; $\be(\alpha_j\mid 1,1)$ is a vaguer prior belief.

Marginalising $G$ yields the same generalised P\'olya urn as in \cite{pitman:95}, now on $\Xi$:
\begin{equation}
\label{eq:prior}
f(\bxi_i\mid\bxi_{-i})
=\frac{b+am_i}{b+n-1}g_0(\bxi_i)
+\sum_{j=1}^{m_i}\frac{n_{i,j}^*-a}{b+n-1}\delta_{\bxi_{i,j}^*}(\bxi_i),
\end{equation}
with the obvious notation for the $m_i$ distinct values $\bxi_{i,j}^*$ and frequencies $n_{i,j}^*$ in $\bxi_{-i}$.

\begin{proposition}[Mixture Kendall tau]
\label{prop:taumix}
Let $C(U,V)$ the marginal bivariate copula corresponding to mixture model \eqref{eq:mixC} with $\bxi=(\theta,\alpha,\beta)$. Kendall's tau for the pair $(U,V)$ is
\begin{equation}
\label{eq:taumix}
\tau=4\sum_{k=1}^\infty\sum_{j=1}^\infty \omega_k\omega_j\iint_{[0,1]^2}C_{\bxi_k}(u,v)c_{\bxi_j}(u,v)\d u\d v-1.
\end{equation}
In general, the vector $(\tau_{jl})_{j<l}$ for the pairs $(U_j,U_l)$ is not constrained to have equal entries.
\end{proposition}

\begin{proof}
Substituting both the mixture copula \eqref{eq:mixC} and its mixture density into \eqref{eq:tau}, and then expanding their product, gives \eqref{eq:taumix}.  The cross-component terms with $k\ne j$ show, in particular, that Kendall's tau is not obtained by simply averaging the component-specific tau's.
\end{proof}

The result in Proposition \ref{prop:taumix} is a correction of equation (9) in  \cite{pan&al:25}. In Section~\ref{ssec:fitting} we will provide a way of approximating \eqref{eq:taumix} using the predictive mixture. 

\subsection{Posterior distributions}
\label{ssec:posterior}

The full posterior conditional distribution of $\bxi_i=(\theta_i,\balpha_i)$ is
\[
f(\bxi_i\mid\bu,\bxi_{-i})
\propto
(b+am_i)\,c_{\bxi_i}(\bu_i)\,g_0(\bxi_i)
+\sum_{j=1}^{m_i}(n_{i,j}^*-a)\,c_{\bxi_i}(\bu_i)\,\delta_{\bxi_{i,j}^*}(\bxi_i).
\]
The likelihood is not conjugate for $g_0$, so we sample with a Gibbs sampler \citep{smith&roberts:93} using Algorithm~8 of \cite{neal:00}, as in \cite{pan&al:25}, with the atom now equal to $\bxi$ rather than $\theta$. Specifically,

\begin{enumerate}
\item[(i)] For each $i=1,\ldots,n$, draw $r$ auxiliary values $\bxi^\star=\{\bxi_{m_i+1}^\star,\ldots,\bxi_{m_i+r}^\star\}$ from $g_0$.
\item[(ii)] Draw $\bxi_i$ from the discrete distribution on $\{\bxi_{i,1}^*,\ldots,\bxi_{i,m_i}^*\}\cup\bxi^\star$ with masses proportional to $(n_{i,j}^*-a)\,c_{\bxi_{i,j}^*}(\bu_i)$ on the existing atoms and to $\{(b+am_i)/r\}\,c_{\bxi_j^\star}(\bu_i)$ on the auxiliary atoms.
\item[(iii)] Let $(\bxi_1^*,\ldots,\bxi_m^*)$ be the distinct values in $\{\bxi_i\}_{i=1}^n$.  Resample each $\bxi_\ell^*=(\theta_\ell^*,\balpha_\ell^*)$ from
\[
f(\bxi_\ell\mid c.c.)
\propto
g_0(\bxi_\ell)\prod_{\{i:\bxi_i=\bxi_\ell^*\}} c_{\bxi_\ell}(\bu_i)
\]
by two Metropolis--Hastings random-walk steps, constrained to $\Xi$, where $c.c.$ stands for clustering configuration. First, $\theta_\ell$ is updated with a uniform proposal on $(\theta_\ell^*-\kappa_\theta,\theta_\ell^*+\kappa_\theta)\cap\Theta$, whose acceptance probability includes the ratio of the (truncated) proposal widths. Second, with $\theta_\ell$ held fixed, each $\alpha_{\ell,j}$ is updated on the logit scale,
\[
\mathrm{logit}(\alpha_{\ell,j}')=\mathrm{logit}(\alpha_{\ell,j}^*)+\varepsilon_j,\qquad\varepsilon_j\sim\no(0,\kappa_\alpha^2),
\]
where $\kappa_\alpha$ is a standard deviation. The Gaussian increment is symmetric on the logit scale, so the proposal density on $\alpha$ carries the logistic Jacobian and the vector $\balpha_\ell'$ is accepted with probability
\[
\min\Biggl\{1,\;
\frac{f(\theta_\ell^*,\balpha_\ell'\mid c.c.)}{f(\theta_\ell^*,\balpha_\ell^*\mid c.c.)}
\prod_{j=1}^d\frac{\alpha_{\ell,j}'(1-\alpha_{\ell,j}')}{\alpha_{\ell,j}^*(1-\alpha_{\ell,j}^*)}\Biggr\}.
\]
\end{enumerate}

The hyperparameters $(a,b)$ are assigned a beta prior for $a$ and a shifted-gamma prior for $b$ given $a$ of the form
\[
f(a,b)=\ga(b+a\mid c_b,d_b)\,\be(a\mid c_a,d_a).
\]
This is updated by the exchangeable partition probability function of \cite{pitman:95}. The corresponding posterior conditional distributions are sampled using MH steps. For $a$ we use the proposal $a'\sim\un(\max\{0,a-\kappa_a\},\min\{a+\kappa_a,1\})$, with the ratio of proposal widths in the acceptance probability, and for $b$ a random walk on the log scale, $\log(b'+a)=\log(b+a)+\varepsilon$ with $\varepsilon\sim\no(0,\kappa_b^2)$, so that $b'+a>0$ automatically.

The tuning parameters $(\kappa_\theta,\kappa_\alpha,\kappa_a,\kappa_b)$, started at $(0.1,1,1,1)$, are adapted \emph{only during the burn-in period}, in batches of 50 iterations, towards the acceptance window $[0.3,0.4]$ \citep{roberts&rosenthal:09}. After batch $j=1,2,\ldots$ each tuning parameter $\kappa$ is multiplied by $\exp(\pm\delta_j)$ with $\delta_j=\min\{0.05,\,j^{-0.6}\}$, where the sign is positive if the acceptance rate $AR^{(j)}$ of the corresponding update in that batch exceeds the target $0.35$ (the centre of the window) and negative otherwise; the acceptance counters are then reset. The adaptation step $\delta_j$ vanishes as $j$ grows, and it stops altogether at the end of burn-in. Hence the retained draws come from a Markov chain with fixed transition kernel (finite adaptation), rather than from a chain that keeps adapting until the end of the run, and no additional ergodicity conditions for adaptive MCMC are required.

\subsection{Fitting measure, pairwise tau, and clustering}
\label{ssec:fitting}

The log pseudo-marginal likelihood $\mathrm{LPML}=\sum_{i=1}^n\log(\mathrm{CPO}_i)$ is computed using a Monte Carlo approximation of the conditional predictive ordinate $\mathrm{CPO}_i=f(\bu_i\mid\bu_{-i})$ \citep{geisser&eddy:79,nieto&contreras:14} as 
\[
\widehat{\mathrm{CPO}}_i
=\left[\frac{1}{L}\sum_{l=1}^L
\left\{c_{\bxi_i^{(l)}}(\bu_i)\right\}^{-1}\right]^{-1},
\]
where $\bxi_i^{(l)}=(\theta_i^{(l)},\balpha_i^{(l)})$ is the atom to which observation $i$ is allocated at retained draw $l$. The same estimator is used for all models and data sets in Section~\ref{sec:numerical}, so that our LPML values are directly comparable with those reported by \cite{pan&al:25}. As a check, we also computed the harmonic mean of the Poisson--Dirichlet predictive densities, which uses the whole mixture; for the 102 simulation fits it differs from the estimator we use by less than 3 LPML units in 99 fits (see the Github repository).

Pairwise Kendall tau's are estimated from simulations from the complete posterior predictive mixture.  This retains the cross-component terms in \eqref{eq:taumix}. A specific clustering is obtained from the posterior similarity matrix by the SALSO search of \cite{dahl&al:22} with Binder loss; given the chosen partition we report for each cluster $\ell$, the posterior of $(\theta_\ell^*,\balpha_\ell^*)$ together with its specific weight $(n_\ell^*-a)/(b+n)$.

The cluster-specific vector $\balpha_\ell^*$ is the new inferential object. Coordinates with $\alpha_{\ell,j}^*$ near one are those through which cluster $\ell$ transmits Archimedean dependence; coordinates with $\alpha_{\ell,j}^*$ near zero are those that cluster $\ell$ treats as independent of the rest.

\section{Numerical experiments}
\label{sec:numerical}

All programs used in this section, together with the configuration files, seeds and scripts, are available at 
\url{https://github.com/rcraiu/NPB_Assymetric_Archimedean}. We
translated, with the help of Cursor \citep{cursor2026},  some of our subroutines written in R into C++ via the Rcpp framework.

The experiments below extend those of \cite{pan&al:25}.  Throughout, the centering measure \eqref{eq:G0} is defined by $c_\theta=4$, $d_\theta=1$ for Clayton, Gumbel and Joe, and $\mu_\theta=0$ and $\lambda_\theta=4$ for Frank. Additionally, $c_\alpha=d_\alpha=1$ unless otherwise stated, and for the prior on $(a,b)$ we take $c_a=c_b=1$ together with $d_a=d_b=20$ that favours a small number of components. MCMC is run for $15{,}000$ iterations with a burn-in of $5{,}000$ and a thinning of $5$.  Simulations~1, 3 and 4 each use one data set generated from the seed recorded in the reproducibility files, whereas Simulation~2 uses 20 seeded data sets; in all cases the competing kernels and the exchangeable mixture of \cite{pan&al:25} (the current model with $\balpha=\bone$) are fitted to identical observations.

\subsection{Simulation 1: a single asymmetric component}

We generated $n\in\{200,500\}$ observations from a bivariate Khoudraji--Clayton copula with $\theta=5$ and $(\alpha,\beta)\in\{(1,1),(0.8,0.3),(0.4,0.4),(0,0)\}$.  The four choices are the exchangeable Clayton copula, an asymmetric copula, a symmetric but diluted copula, and independence.  We fitted \eqref{eq:bnpmix} and the exchangeable mixture with Clayton, Frank and Gumbel kernels.

Table~\ref{tab:sim1} gives the results. When the truth is exchangeable, the extended model sends both Khoudraji coordinates towards one: at $n=500$ the posterior means range from $(0.96,0.96)$ to $(0.99,0.99)$. The exchangeable model has the larger LPML for all three kernels, as expected when the extra parameters are unnecessary. For the asymmetric truth, the distinction becomes clear at $n=500$. The Clayton fit recovers $(\hat\alpha,\hat\beta)=(0.86,0.35)$ and raises LPML from $27.15$ to $44.32$; the Frank and Gumbel gains are $10.66$ and $5.22$, respectively. The corresponding posterior-predictive tau under Clayton is $0.230$, close to the empirical $0.239$. At $n=200$ the sample tau is only $0.093$, and consequently the direction is weakly identified.

For the diluted symmetric truth, extended Clayton improves LPML by $2.27$ at $n=500$ and estimates $(0.52,0.41)$; the Frank gain is small and Gumbel does not improve.  Under independence, all posterior-predictive tau's are close to zero and the broad posteriors of both $\balpha$ and $\theta$ reflect their weak identification. Thus, the experiment supports recovery of substantial asymmetry, but also shows that the Khoudraji coordinates should not be interpreted componentwise when dependence is nearly absent.

\begin{table}[p]
\centering
\footnotesize
\begin{tabular}{cclrrrrrr}
\hline\hline\\[-2mm]
True $\balpha$ & $n$ & Kernel & \multicolumn{2}{c}{LPML} & $\hat\theta$ & $\hat\balpha$ & $\tilde\tau$ & $\hat\tau$ \\
\cline{4-5}
 & & & Exch & Khou & & & & \\
\hline
$(1,1)$ & 200 & CLA & 208.15 & 206.38 & 5.90 & (0.98,0.99) & 0.758 & 0.725 \\
$(1,1)$ & 200 & FRA & 167.23 & 166.28 & 11.89 & (0.99,0.97) & 0.758 & 0.684 \\
$(1,1)$ & 200 & GUM & 158.22 & 142.58 & 5.07 & (0.92,0.91) & 0.758 & 0.660 \\
$(0.8,0.3)$ & 200 & CLA & -0.21 & 0.58 & 3.93 & (0.41,0.34) & 0.093 & 0.085 \\
$(0.8,0.3)$ & 200 & FRA & 1.19 & 0.64 & 1.42 & (0.55,0.54) & 0.093 & 0.060 \\
$(0.8,0.3)$ & 200 & GUM & 0.36 & -0.87 & 4.05 & (0.21,0.22) & 0.093 & 0.060 \\
$(0.4,0.4)$ & 200 & CLA & 1.05 & 3.45 & 4.47 & (0.48,0.37) & 0.142 & 0.138 \\
$(0.4,0.4)$ & 200 & FRA & 3.75 & 2.92 & 2.11 & (0.67,0.59) & 0.142 & 0.113 \\
$(0.4,0.4)$ & 200 & GUM & 1.62 & 0.66 & 3.76 & (0.48,0.18) & 0.142 & 0.068 \\
$(0,0)$ & 200 & CLA & -3.48 & -0.83 & 3.62 & (0.22,0.28) & -0.053 & 0.010 \\
$(0,0)$ & 200 & FRA & -0.45 & 0.18 & -1.18 & (0.45,0.45) & -0.053 & -0.022 \\
$(0,0)$ & 200 & GUM & -3.52 & -1.75 & 4.75 & (0.14,0.15) & -0.053 & 0.024 \\ \hline
$(1,1)$ & 500 & CLA & 491.91 & 490.14 & 4.78 & (0.99,0.99) & 0.733 & 0.690 \\
$(1,1)$ & 500 & FRA & 374.46 & 374.13 & 11.40 & (0.98,0.98) & 0.733 & 0.679 \\
$(1,1)$ & 500 & GUM & 342.88 & 319.76 & 4.42 & (0.96,0.96) & 0.733 & 0.612 \\
$(0.8,0.3)$ & 500 & CLA & 27.15 & 44.32 & 5.08 & (0.86,0.35) & 0.239 & 0.230 \\
$(0.8,0.3)$ & 500 & FRA & 30.28 & 40.94 & 5.39 & (0.94,0.43) & 0.239 & 0.230 \\
$(0.8,0.3)$ & 500 & GUM & 22.01 & 27.23 & 2.51 & (0.84,0.31) & 0.239 & 0.190 \\
$(0.4,0.4)$ & 500 & CLA & 18.98 & 21.25 & 5.14 & (0.52,0.41) & 0.187 & 0.190 \\
$(0.4,0.4)$ & 500 & FRA & 19.37 & 19.91 & 3.23 & (0.74,0.62) & 0.187 & 0.167 \\
$(0.4,0.4)$ & 500 & GUM & 20.82 & 19.20 & 1.82 & (0.55,0.43) & 0.187 & 0.161 \\
$(0,0)$ & 500 & CLA & -2.60 & -0.60 & 3.56 & (0.37,0.18) & -0.012 & 0.021 \\
$(0,0)$ & 500 & FRA & -0.84 & -0.45 & -0.30 & (0.40,0.42) & -0.012 & -0.002 \\
$(0,0)$ & 500 & GUM & -2.64 & -1.46 & 4.56 & (0.21,0.20) & -0.012 & 0.019 \\
\hline\hline
\end{tabular}
\caption{Simulation 1: data generated from a single Clayton-Khoudraji model with $\theta=5$ and the true $\balpha$ given in the first column. LPML for the exchangeable and Khoudraji mixtures; occupancy-weighted posterior means under the Khoudraji model; sample tau $\tilde\tau$; and posterior-predictive tau $\hat\tau$.}
\label{tab:sim1}
\end{table}

\subsection{Simulation 2: a mixture of mirror-asymmetric components}

This experiment asks whether the Khoundraji coordinates $\balpha$ allow the mixture to separate clusters that an exchangeable kernel can not separate. We generated $n=500$ observations from the equally weighted two-component Clayton--Khoudraji mixture
\[c(u,v)=
\tfrac12\,c_{\theta,1,0.65}(u,v)
+\tfrac12\,c_{\theta,0.65,1}(u,v),\qquad \theta=50 .
\]
The components are mirror images: they share $\theta$ and hence the same Kendall tau ($\approx0.63$) and marginal behaviour, and differ only in the direction of asymmetry.  Every exchangeable Archimedean kernel is invariant under the swap $(u,v)\mapsto(v,u)$, so any esxchangeable mixture assigns a point and its reflection the same density; the Bayes allocation rule based on the best exchangeable pair of kernels therefore has accuracy $0.50$ by symmetry, whereas the allocation rule using the true Khoudraji parameters has accuracy about $0.87$ and adjusted Rand index (ARI) about $0.55$. The ARI is modest because the two components overlap near the diagonal and in the upper corner.

The prior on $\theta$ is $\ga(4,0.08)$ (mean $50$) so that the true value is not in the far tail of the default $\ga(4,1)$, and the prior on $b+a$ is $\ga(2,10)$; all other settings are as above. We report the following recovery summaries against the generating labels: the posterior-mean ARI, i.e. the ARI of each retained allocation vector, averaged over draws; the ARI of the SALSO point estimate under Binder loss, restricted to at most two clusters, and under VI loss. 
The whole experiment is repeated on 20 independent seeded data sets. The oracle fit is that of the Bayes allocation rule using the true parameters. 

\paragraph{A single data set.} The featured data set (replicate~1) has 252 and 248 observations from the two components and overall sample tau $0.608$.  Table~\ref{tab:sim2} shows that the extended model raises LPML from $166.58$ to $527.10$.  The exchangeable fit places all observations in one group under Binder and VI loss (ARI $0.00$), while its posterior-predictive tau of $0.477$ is well below the sample value. The Khoudraji fit gives ARI $0.52$ under unrestricted Binder loss, which selects 7 groups (two large groups and small fragments), and $0.54$ when Binder loss is restricted to at most two groups, close to the oracle benchmark. These two groups have sizes 246 and 254, and occupancy-weighted posterior means $(\hat\theta,\hat\balpha)=(51.58,0.92,0.69)$ and $(54.56,0.68,0.91)$ (Table~\ref{tab:sim2-clusters}): the fit recovers both $\theta$ and the mirrored directions of asymmetry. Figure~\ref{fig:sim2} shows that the Khoudraji predictive law reproduces the two ridges on either side of the diagonal, whereas the exchangeable predictive law merges them into a single cloud, and that the PSM has the expected two-block structure. The VI summary, in contrast, merges the Khoudraji fit into a single group (ARI $0.00$); this sensitivity to the loss function is reported rather than suppressed, and it reflects genuine posterior uncertainty: the mean posterior co-clustering probability is $0.67$ for pairs from the same true component and $0.33$ for pairs from different components.

\begin{table}[ht!]
\centering
\small
\begin{tabular}{lrrrrrrrrr}
\hline\hline\\[-4mm]
Model & LPML & $\hat\tau$ & $\hat{m}$ & $\widehat{\rm ARI}$ & $K_B$ & ARI$_B$ & ARI$_{B,2}$ & $K_V$ & ARI$_V$ \\
\hline
Exch & 166.58 & 0.477 & 1.99 & 0.00 & 1 & 0.00 & 0.00 & 1 & 0.00 \\
Khou & 527.10 & 0.551 & 2.60 & 0.34 & 7 & 0.52 & 0.54 & 1 & 0.00 \\
\hline\hline
\end{tabular}
\caption{Simulation 2, replicate 1: data generated from the equally weighted mixture of two Clayton-Khoudraji copulas with $\theta=50$ and $(\alpha,\beta)=(1,0.65)$ and $(0.65,1)$ ($n=500$). Posterior summaries: LPML, Kendall's tau $\hat\tau$, number of groups $\hat{m}$, and ARI; $K_B$ and $K_V$ are the numbers of groups selected by unrestricted Binder and VI loss with their respective ARI$_B$ and ARI$_V$; ARI$_{B,2}$ is the Binder result constrained to at most two groups.}
\label{tab:sim2}
\end{table}

\begin{table}[ht!]
\centering
\begin{tabular}{lrrrr}
\hline\hline\\[-4mm]
Partition & Cluster & Size & $\hat\theta$ & $\hat\balpha$ \\
\hline
Binder $K\leq2$ & 1 & 246 & 51.58 & (0.92,0.69) \\
Binder $K\leq2$ & 2 & 254 & 54.56 & (0.68,0.91) \\
VI & 1 & 500 & 53.09 & (0.80,0.80) \\
\hline\hline
\end{tabular}
\caption{Simulation 2, replicate 1: data generated from the equally weighted mixture of two Clayton-Khoudraji copulas with $\theta=50$ and $(\alpha,\beta)=(1,0.65)$ and $(0.65,1)$ ($n=500$). Cluster sizes and occupancy-weighted parameter summaries for the Khoudraji Clayton fit under the two-cluster Binder and
unrestricted VI summaries.}
\label{tab:sim2-clusters}
\end{table}

\begin{figure}[ht!]
\centering
\includegraphics[width=\textwidth]{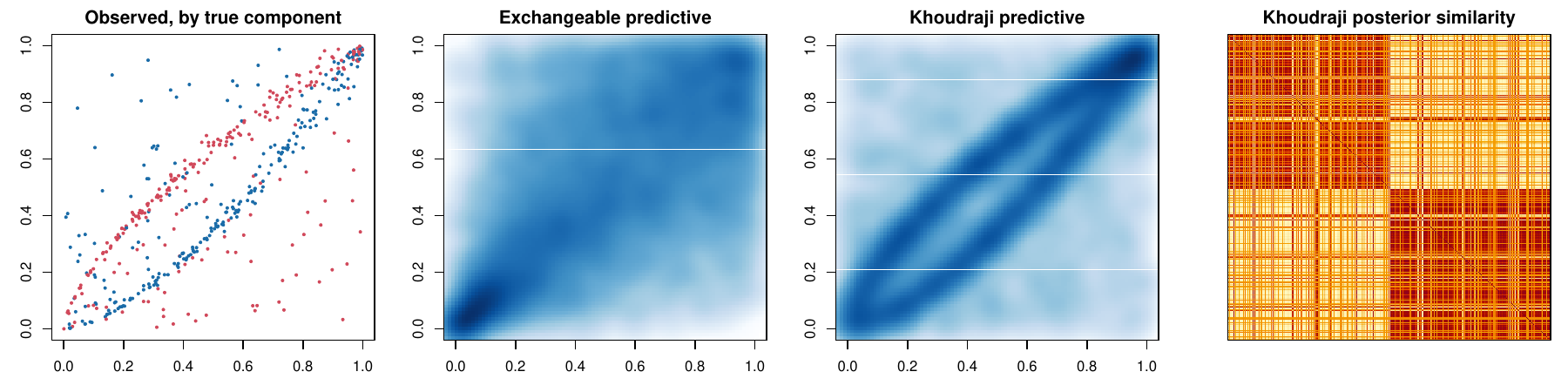}
\caption{Simulation 2, replicate 1: data generated from the equally weighted mixture of two Clayton-Khoudraji copulas with $\theta=50$ and $(\alpha,\beta)=(1,0.65)$ and $(0.65,1)$ ($n=500$). Observed pseudo-observations coloured by true component, posterior-predictive contours for the exchangeable and Khoudraji Clayton mixtures, and the Khoudraji posterior similarity matrix with observations ordered by true component.}
\label{fig:sim2}
\end{figure}

\paragraph{Twenty data sets.}  Table~\ref{tab:sim2-replicates} and Figure~\ref{fig:sim2-replicates} summarize all replicates. The extended model has the larger LPML in every data set, around $348.8$ on average (range $296.7$ to $416.2$).  Cluster recovery tells the same story: the exchangeable model has ARI essentially $0$ in every data set (it selects a single group in 19 of 20 and two groups once), whereas the extended model attains a mean ARI of $0.55$ (standard deviation $0.05$) under both unrestricted Binder loss and the PSM cut, matching the oracle benchmark of $0.55$ (standard deviation $0.04$); the average difference from the oracle is $-0.002$.  The Binder $K\le2$ summary returns two groups with opposite orientation ($\hat\alpha>\hat\beta$ in one group and $\hat\alpha<\hat\beta$ in the other) in all 20 Khoudraji fits and in none of the exchangeable fits. The posterior-mean ARI is lower ($0.37$), reflecting the genuine posterior uncertainty about the allocation of observations near the diagonal and in the upper corner. Unrestricted Binder loss selects between 2 and 11 groups (mean $7.7$), but the additional groups are small fragments that do not change the ARI, and VI loss merges the extended fit into a single group in all 20 data sets. We therefore recommend Binder (possibly restricted) or PSM-based summaries for reporting cluster recovery here.

\begin{table}[ht!]
\centering
\small
\begin{tabular}{lcc}
\hline\hline
Summary over 20 data sets & Khou & Exch \\
\hline
LPML & 485.4 (31.5) & 136.6 (12.7) \\
$\hat{m}$ & 2.76 (0.13) & 1.61 (0.37) \\
$K_B$, Binder & 7.7 (2.1) & 1.1 (0.2) \\
ARI, Binder & 0.55 (0.05) & 0.00 (0.00) \\
ARI, Binder ($K\leq2$) & 0.55 (0.05) & 0.00 (0.00) \\
ARI, VI & 0.00 (0.00) & 0.00 (0.00) \\
ARI, PSM cut at $K=2$ & 0.55 (0.04) & 0.00 (0.00) \\
$\widehat{\rm ARI}$ & 0.37 (0.03) & 0.00 (0.00) \\
PSM within $-$ between & 0.37 (0.03) & 0.00 (0.00) \\
Oracle ARI & \multicolumn{2}{c}{0.55 (0.04)} \\
Khou beats Exch in LPML / ARI$_{B,2}$ & \multicolumn{2}{c}{20 / 20 of 20} \\
True clusters found, Binder ($K\leq2$) & 20 of 20 & 0 of 20 \\
\hline\hline
\end{tabular}
\caption{Simulation 2 over 20 seeded data sets, each generated from the equally weighted mixture of two Clayton-Khoudraji copulas with $\theta=50$ and $(\alpha,\beta)=(1,0.65)$ and $(0.65,1)$ ($n=500$). Mean (standard deviation) of each summary under the Khoudraji and exchangeable models.}
\label{tab:sim2-replicates}
\end{table}

\begin{figure}[ht!]
\centering
\includegraphics[width=\textwidth]{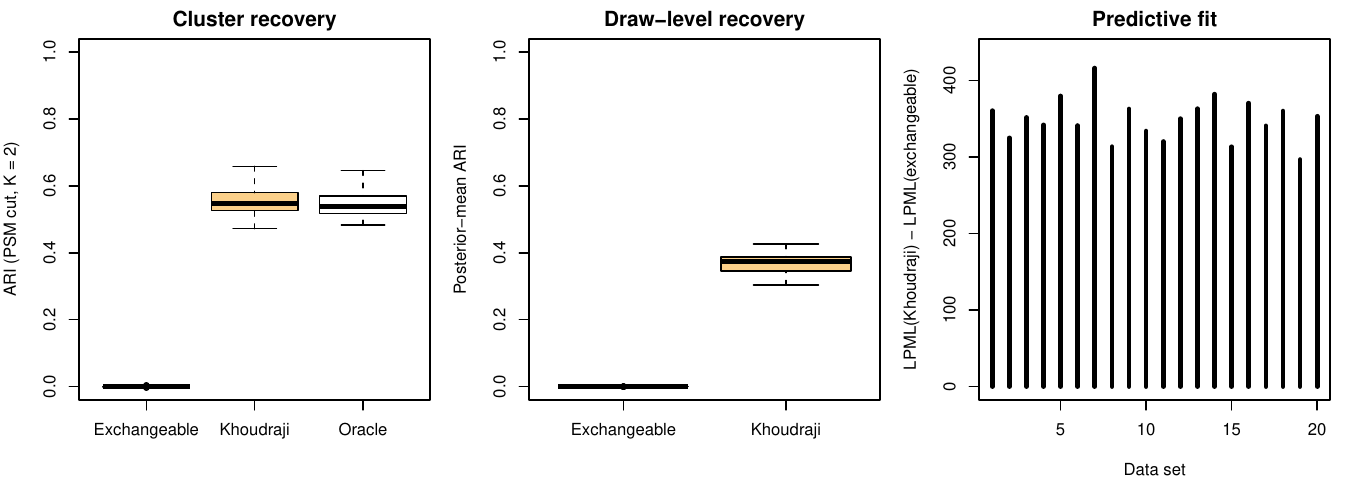}
\caption{Simulation 2 over 20 seeded data sets, each generated from the equally weighted mixture of two Clayton-Khoudraji copulas with $\theta=50$ and $(\alpha,\beta)=(1,0.65)$ and $(0.65,1)$ ($n=500$). Left: ARI of the PSM cut at $K=2$ for the exchangeable and Khoudraji fits and for the oracle allocation rule.  Centre: posterior-mean ARI.  Right: LPML difference, Khoudraji minus exchangeable.}
\label{fig:sim2-replicates}
\end{figure}

\subsection{Simulation 3: unequal pairwise tau's in dimension three}

We generate $n=1000$ observations from a three-dimensional Khoudraji--Gumbel copula with $\theta=4$ and $\balpha=(1,0.7,0.2)$. The three sample pairwise tau's are $(0.553,0.188,0.182)$. The exchangeable mixture compromises at approximately $(0.297,0.293,0.294)$, whereas the extended model gives $(0.560,0.196,0.184)$ (Table~\ref{tab:sim3}). It also recovers $\balpha$ almost exactly: its posterior means are $(0.995,0.702,0.214)$, with marginal 95\% credible intervals $(0.985,1.000)$, $(0.672,0.732)$ and $(0.200,0.229)$, respectively. The resulting LPML gain is $470.34$. Clearly, the Khoudraji layer provides a dependence structure that no mixture of exchangeable Ar\-chi\-me\-dean kernels can reproduce.

\begin{table}[ht!]
\centering
\begin{tabular}{lrrrrr}
\hline\hline
Model & LPML & $\hat\balpha$ & $\tau_{12}$ & $\tau_{13}$ & $\tau_{23}$ \\
\hline
Sample & --- & --- & 0.553 & 0.188 & 0.182 \\
Exch & 317.17 & (1.00,1.00,1.00) & 0.297 & 0.293 & 0.294 \\
Khou & 787.51 & (1.00,0.70,0.21) & 0.560 & 0.196 & 0.184 \\
\hline\hline
\end{tabular}
\caption{Simulation 3: $n=1000$ observations generated from a three-dimensional Gumbel-Khoudraji copula with $\theta=4$ and $\balpha=(1,0.7,0.2)$. LPML, occupancy-weighted posterior mean of $\balpha$, and posterior-predictive pairwise Kendall tau's.}
\label{tab:sim3}
\end{table}

\subsection{Simulation 4: Gaussian copulas}

We revisit Simulation~2 of \cite{pan&al:25}: a single Gaussian copula with $\rho=0.7$ and an equally weighted two-component Gaussian mixture with $\rho\in\{-0.7,0.7\}$. Table~\ref{tab:sim4} and Figure~\ref{fig:sim4} show that the additional coordinates do not generally solve this misspecification. For the single Gaussian copula, Khoudraji Clayton gains $16.34$ LPML units over exchangeable Clayton, but exchangeable Frank remains overall best (LPML $153.69$) and its predictive tau $0.493$ closely matches the empirical $0.491$. The extended Frank and Gumbel fits place $\balpha$ close to $\bone$ and lose $1.43$ and $2.41$ LPML units, respectively.

For the two-component Gaussian mixture the sample tau is only $0.038$, despite strong dependence within each component. Frank mixtures capture the opposing signs through two principal clusters and dominate the one-sided Clayton and Gumbel families. Exchangeable Frank attains LPML $54.19$, compared with $53.63$ for extended Frank; its predictive tau is $0.005$. Khoudraji Clayton improves upon exchangeable Clayton but remains more than $40$ LPML units behind Frank. Here, the limitation is the generator's ability to combine positive and negative dependence, not exchangeability, and the Khoudraji coordinates cannot compensate for it.

\begin{table}[ht!]
\centering
\small
\begin{tabular}{llrrrrrr}
\hline\hline
Truth & Kernel & \multicolumn{2}{c}{LPML} & $\Delta$ & $\hat\balpha$ & \multicolumn{2}{c}{$\hat\tau$} \\
\cline{3-4}\cline{7-8}
 & & Exch & Khou & & & Exch & Khou \\
\hline
$\rho=0.7$ & CLA & 126.72 & 143.06 & +16.34 & (0.85,0.78) & 0.428 & 0.427 \\
$\rho=0.7$ & FRA & 153.69 & 152.26 & -1.43 & (0.97,0.96) & 0.493 & 0.480 \\
$\rho=0.7$ & GUM & 141.51 & 139.09 & -2.41 & (0.97,0.97) & 0.459 & 0.438 \\
$\rho=\{-0.7,0.7\}$ & CLA & 6.92 & 12.89 & +5.97 & (0.37,0.33) & 0.093 & 0.117 \\
$\rho=\{-0.7,0.7\}$ & FRA & 54.19 & 53.63 & -0.56 & (0.93,0.92) & 0.005 & 0.000 \\
$\rho=\{-0.7,0.7\}$ & GUM & 15.52 & 10.71 & -4.80 & (0.19,0.38) & 0.116 & 0.098 \\
\hline\hline
\end{tabular}
\caption{Simulation 4: data generated from a Gaussian copula, either a single component with $\rho=0.7$ or an equally weighted mixture of two Gaussian copulas with $\rho=-0.7$ and $\rho=0.7$ (first column; $n=500$). LPML and posterior-predictive tau for exchangeable and Khoudraji models. Here $\Delta=\mathrm{LPML}_{\rm Khou}-\mathrm{LPML}_{\rm Exch}$ and $\hat\balpha$ is reported for the Khoudraji fit.}
\label{tab:sim4}
\end{table}

\begin{figure}[p]
\centering
\includegraphics[width=.98\textwidth]{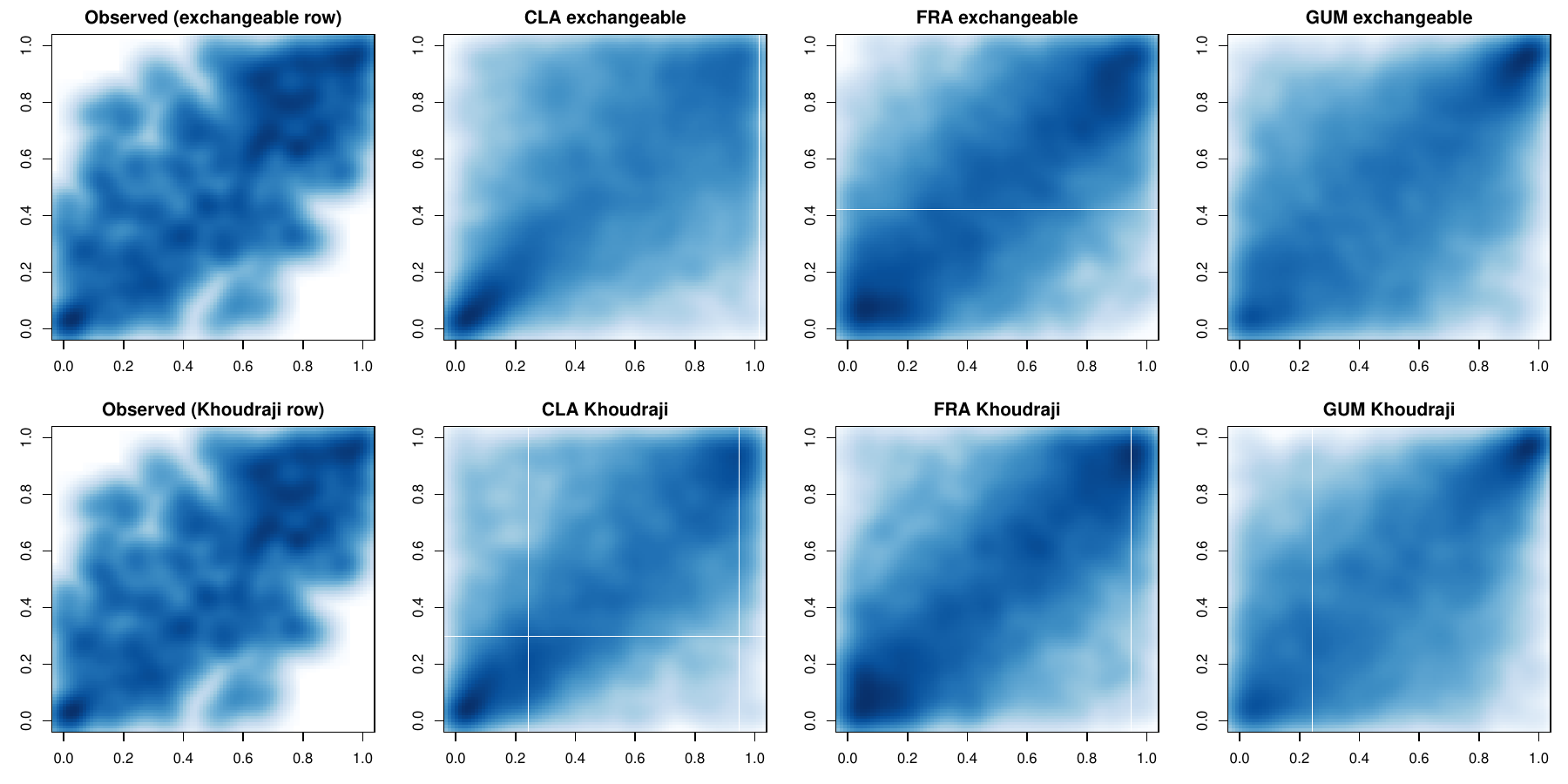}\\[2mm]
\includegraphics[width=.98\textwidth]{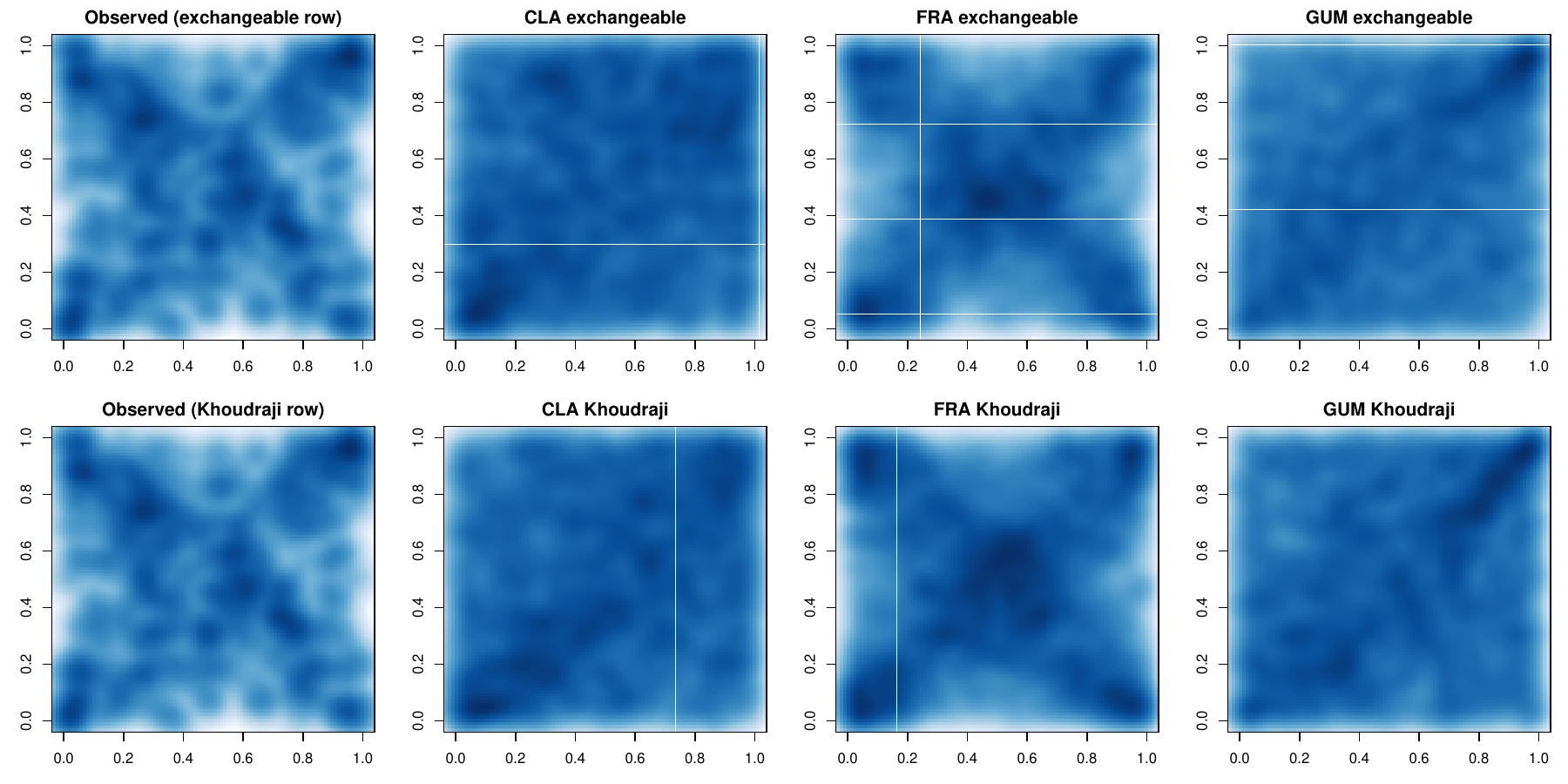}
\caption{Simulation 4: data generated from a Gaussian copula with $\rho=0.7$ (top two rows) or from an equally weighted mixture of Gaussian copulas with $\rho=-0.7$ and $\rho=0.7$ (bottom two rows), $n=500$. Observed samples (first column) and posterior-predictive contours; within each pair of rows, exchangeable fits appear above Khoudraji fits.}
\label{fig:sim4}
\end{figure}

\subsection{Occupancy and red-wine data}
\label{sec:real}

We refit the two data sets of \cite{pan&al:25} with the Khoudraji mixture \eqref{eq:bnpmix}, using the same modified rank transform, the same centering measure $g_0$ and the same MCMC length: $15{,}000$ iterations, burn-in $5{,}000$, thinning $5$, with the tuning constants adapted during burn-in only (Section~\ref{ssec:posterior}). Copula densities are evaluated in compiled C++, and LPML is computed exactly as in the simulations (Section~\ref{ssec:fitting}), which is the estimator of \cite{pan&al:25}. The exchangeable model $\balpha=\bone$ is the Poisson--Dirichlet mixture of \cite{pan&al:25}; comparing the two LPML values isolates the contribution of the Khoudraji device. 

\paragraph{Occupancy.}
We take occupancy measurements of \cite{candanedo2016accurate}: five-minute averages of CO$_2$ and humidity ratio on 5 February 2015, first-difference, and rank-transform ($n=287$). The sample Kendall tau is $0.481$, as seen in Table~\ref{tab:lpml-occ} which reports the LPML for the exchangeable and Khoudraji models for 4 copula families. For Gumbel, the exchangeable refit ($115.5$) and the Khoudraji model ($113.9$) are practically indistinguishable: the posterior of $(\alpha_1,\alpha_2)$ concentrates near one, with means $(0.91,0.92)$. The star-shaped scatter is therefore exchangeable; asymmetry is not needed, and a single Gumbel component with $\theta\approx 2.4$ remains the preferred description. %
In no occupancy model does the Khoudraji layer improve on the exchangeable fit.

\begin{table}[ht!]
\centering
\begin{tabular}{lrr} \hline\hline \\[-4mm]
Model &  Exch & Khoudraji \\[1mm] \hline \\[-3mm]
AMH &  $76.7$ & $74.8$ \\
CLA &  $92.4$ & $81.5$ \\
FRA &  $91.2$ & $87.2$ \\
GUM &  $115.5$ & $113.9$ \\
\hline\hline
\end{tabular}
\caption{Occupancy data (real data: the generating model is unknown). LPML for the  exchangeable model ($\balpha=\bone$), and the Khoudraji mixture. }
\label{tab:lpml-occ}
\end{table}

\paragraph{Red wine.}
The three-dimensional red-wine data of \cite{cortez2009modeling} has size $n=1{,}599$ and three variables: fixed acidity, citric acid, density. The sample pairwise tau's are: $\tilde\tau_{12}=0.484$, $\tilde\tau_{13}=0.457$ and $\tilde\tau_{23}=0.245$. An exchangeable Archimedean mixture must use a single $\tau$ for all three pairs. Table~\ref{tab:lpml-wine} shows  our exchangeable and Khoudraji fits.  Allowing $\balpha$ to vary raises LPML, relative to the exchangeable refits, by $158$ for Gumbel, $133$ for Frank and $204$ for Clayton. The Gumbel Khoudraji mixture is the preferred model. Its posterior mean of $\balpha$ is $(0.93,0.65,0.66)$: fixed acidity stays almost fully in the Archimedean factor, while citric acid and density are diluted, which is exactly the pair with the smaller empirical tau. The posterior mode of the number of components is $4$, against $2$ in the exchangeable Gumbel fit.

\begin{table}[ht!]
\centering
\begin{tabular}{lrrl} \hline\hline \\[-3mm]
Model &  Exch & Khoudraji & $\hat\balpha$ \\[1mm] \hline \\[-3mm]
CLA &  $607.9$ & $811.9$ & $(0.75,0.58,0.60)$ \\
FRA &  $679.6$ & $812.8$ & $(0.94,0.69,0.62)$ \\
GUM &  $716.4$ & $874.1$ & $(0.93,0.65,0.66)$ \\
\hline\hline
\end{tabular}
\caption{Red wine data (real data: the generating model is unknown). LPML for the exchangeable mixture model ($\balpha=\bone$), and the Khoudraji mixture model, together with the posterior mean of $(\alpha_1,\alpha_2,\alpha_3)$ for fixed acidity, citric acid and density. }
\label{tab:lpml-wine}
\end{table}

\section{Concluding remarks}
\label{sec:conclusion}

The Poisson--Dirichlet mixture of Archimedean copulas of \cite{pan&al:25} inherits the exchangeability of its kernel. Replacing that kernel by the Khoudraji transform \eqref{eq:khoudraji} keeps the computational structure of their sampler --- a P\'olya urn on atoms, Neal's Algorithm~8, and the same updates for $(a,b)$ --- and adds a vector $\balpha\in[0,1]^d$ that interpolates, coordinate-wise, between Archimedean dependence and independence. 

Pairwise Kendall tau's are then allowed to differ, lower-tail dependence requires $\alpha_j=\alpha_k=1$ while upper-tail dependence (Gumbel, Joe) survives for $\alpha_j,\alpha_k>0$, and clusters may differ in both the strength $\theta$ and the direction $\balpha$ of association. The exchangeable model is obtained when $\balpha=\bone$, so the two specifications are comparable by LPML.

For the occupancy data, the posterior of $\balpha$ concentrates at $\bone$ and the Khoudraji Gumbel mixture matches the exchangeable fit of \cite{pan&al:25}.  On the red-wine data, where the three pairwise tau's are not equal, the same Gumbel kernel with an extended $\balpha$ improves LPML by about $160$ and loads Archimedean dependence on fixed acidity while diluting the citric-acid--density pair.

Future work will consider mixtures of vines or hierarchical Archimedean copulas \citep{hofert&pham:13,li&al:21}.  

\section*{Acknowledgements}
The first author was supported by NSERC of Canada RGPIN-2024-04506 Discovery Grant. The second author was supported by \textit{Asociaci\'on Mexicana de Cultura, A.C.}

\bibliographystyle{ims}
\bibliography{References}

\end{document}